\documentclass[a4paper,11pt]{article}
\usepackage{fullpage}
\usepackage{amsthm, amsfonts,mathrsfs,amssymb,amsmath,clrscode,multirow,graphicx,amsmath,epic,eepic}
\usepackage{graphicx}

\newcommand{\ket}[1]{\vert #1 \rangle}

\newcommand{\ceil}[1]{\left\lceil #1 \right\rceil}
\newcommand{\mybar}[1]{\lambda}

\newcommand{\mydelta}[1]{\Delta_G(X,#1)}
\newcommand{\E}{\mathcal{E}}

\newcommand{\poly}{\mathrm{poly}}

\newcommand{\Ss}{\mathcal{S}}

\newtheorem{theorem}{Theorem}[section]

\newtheorem{proposition}{Proposition}[section]
\newtheorem{definition}{Definition}[section]

\newtheorem{lemma}{Lemma}[section]
\newenvironment{proof-sketch}{\trivlist\item[]\emph{Brief proof sketch}:}%
{\unskip\nobreak\hskip 1em plus 1fil\nobreak$\Box$
\parfillskip=0pt%
\endtrivlist}
\begin{document}
\author{%
Fran{\c c}ois Le Gall \\ 
     Department of Computer Science\\
     Graduate School of Information Science and Technology\\
    The University of Tokyo\\
    \texttt{legall@is.s.u-tokyo.ac.jp} }
\title{Improved Quantum Algorithm for Triangle Finding via Combinatorial Arguments}
\date{}
\maketitle\vspace{20mm}
\begin{abstract} 
In this paper we present a quantum algorithm solving the triangle finding problem in unweighted graphs with query complexity $\tilde O(n^{5/4})$, where $n$ denotes the number of vertices in the graph. This improves the previous upper bound $O(n^{9/7})=O(n^{1.285...})$ recently obtained by Lee, Magniez and Santha. Our result shows, for the first time, that in the quantum query complexity setting unweighted triangle finding is easier than its edge-weighted version, since for finding an edge-weighted triangle Belovs and Rosmanis proved that any quantum algorithm requires $\Omega(n^{9/7}/\sqrt{\log n})$ queries. Our result also illustrates some limitations of the non-adaptive learning graph approach used to obtain the previous $O(n^{9/7})$ upper bound since, even over unweighted graphs, any quantum algorithm for triangle finding obtained using this approach requires $\Omega(n^{9/7}/\sqrt{\log n})$ queries as well. To bypass the obstacles characterized by these lower bounds, our quantum algorithm uses combinatorial ideas exploiting the graph-theoretic properties of triangle finding, which cannot be used when considering edge-weighted graphs or the non-adaptive learning graph approach.
\end{abstract}
\newpage
\section{Introduction}
\paragraph{The triangle finding problem and its classical complexity.}
Triangle finding is a graph-theoretic problem 
whose complexity is deeply connected to the complexity of 
several other computational tasks in theoretical computer science, 
such as solving path or matrix problems \cite{Bansal+ToC12,Czumaj+SICOMP09,Itai+SICOMP78,PatrascuSTOC10,Vassilevska+SICOMP13,Williams+FOCS10,WilliamsSTOC14}. 
In its standard version 
it asks to find, given an undirected and unweighted graph $G=(V,E)$, 
three vertices $v_1,v_2,v_3\in V$ such that $\{v_1,v_2\}$, $\{v_1,v_3\}$
and $\{v_2,v_3\}$ are edges of the graph. It has been known for a long time that
this problem is not harder than Boolean matrix multiplication \cite{Itai+SICOMP78},
which implies that triangle finding in a graph of $n$ vertices can be solved in
$O(n^{\omega+\varepsilon})$ time for any constant $\varepsilon>0$, where~$\omega$ represents the 
exponent of square matrix multiplication (currently, the best known upper bound on $\omega$ is $\omega<2.3729$, see \cite{LeGall14,WilliamsSTOC12}). This is still the best known upper bound on the classical time complexity of triangle finding.
Recently, Vassilevska Williams and Williams showed a converse reduction~\cite{Williams+FOCS10}: they proved that a subcubic-time algorithm for triangle finding can be used, in a combinatorial way, to construct a subcubic-time algorithm for Boolean matrix multiplication. 

Much progress has furthermore been achieved recently on understanding the classical complexity of weighted versions of 
the triangle finding problem \cite{Czumaj+SICOMP09,PatrascuSTOC10,Vassilevska+SICOMP13,Williams+FOCS10,WilliamsSTOC14}.
In particular, it has been shown that the exact 
node-weighted triangle finding problem, where the goal is to find three vertices 
in a node-weighted graph 
such that the sum of the weights of these three vertices is equal to a given value,
is not harder than matrix multiplication over a field \cite{Czumaj+SICOMP09,Vassilevska+SICOMP13}. 
For the exact edge-weighted triangle finding problem, where the goal is to find three vertices 
$v_1,v_2,v_3$ in a edge-weighted graph 
such that the sum of the weights of $\{v_1,v_2\}$, $\{v_1,v_3\}$ and $\{v_2,v_3\}$ 
is equal to a given value, it has been shown recently that the situation is completely different:
it requires $\Omega(n^{3-\delta})$ time for all $\delta>0$ 
unless the 3SUM problem on~$N$ integers can be solved in $O(N^{2-\delta/6})$ time 
\cite{PatrascuSTOC10,Vassilevska+SICOMP13}, which strongly suggests that the edge-weighted version
of triangle finding is harder than its node-weighted and unweighted versions.

In this paper \emph{triangle finding problem} will always refer to the unweighted version -- this is the version studied in this paper, as in most previous works on quantum algorithms. The words \emph{node-weighted} or \emph{edge-weighted} will be explicitly added when referring to the weighted versions.

\paragraph{Quantum query algorithms for triangle finding.}
Besides the time complexity setting discussed above, problems like triangle finding can also be studied in the query complexity setting.
In the usual model used to describe the query complexity of such problems,
the set of edges $E$ of the graph is unknown but can be accessed through an oracle: given two vertices $u$ and $v$ in $V$, 
one query to the oracle outputs one if $\{u,v\}\in E$ and zero if $\{u,v\}\notin E$. In the quantum query complexity setting, 
one further assume that the oracle can be queried in superposition. One of the main interests of query
 complexity is that, being a restricted model of computation, in many cases one can show lower bounds on the complexity of problems (in both the classical and quantum settings). For instance, it is straightforward to show that the randomized classical query complexity of triangle finding is $\Omega(n^2)$, where $n$ denotes the number of vertices in the graph, 
 which matches the trivial upper bound. 
 In comparison, several better quantum query algorithms have been developed. 
 Indeed, besides its theoretical interest, the triangle finding problem has been one of the main problems 
 that stimulated the development of new techniques in quantum query complexity, and the history of improvement
 of upper bounds on the query complexity of triangle finding parallels the development of general techniques in the 
 quantum complexity setting, as we explain below.

Among the first techniques for constructing quantum algorithms were 
Grover search \cite{GroverSTOC96},
which can be used to obtain a quadratic speedup over classical exhaustive search for any 
unstructured search problem,
 and its variant known as amplitude amplification \cite{Brassard+02}.
Grover search immediately gives, 
when applied to triangle finding as a search over
the space of triples of vertices of the graph, a quantum algorithm 
with query complexity $O(n^{3/2})$.
Using amplitude amplification,
Buhrman et al.~\cite{Buhrman+SICOMP05} showed how to construct a quantum algorithm for triangle finding with query  
complexity $O(n+\sqrt{nm})$ for a graph with~$m$ edges, giving an improvement  for sparse 
graphs. Combining amplitude amplification with clever combinatorial arguments, 
Szegedy \cite{Szegedy03} (see also \cite{Magniez+SICOMP07}) 
constructed a quantum algorithm for triangle finding with query complexity 
$\tilde O(n^{10/7})=\tilde O(n^{1.428...})$.\footnote{In this paper the $\tilde O(\cdot)$
notation removes $\poly(\log n)$ factors.}

The quantum technique that led to the next improvement was the 
concept of quantum walk search developed by Ambainis~\cite{AmbainisSICOMP07}, 
and used originally to construct an optimal quantum algorithm for the element distinctness problem~\cite{AmbainisSICOMP07}.
This new approach, which combines amplitude amplification with a quantum version of random 
walks over Johnson graphs and was later generalized to quantum walks over 
more general graphs \cite{Magniez+SICOMP11,SzegedyFOCS04},
has turned out to be one of the most useful tools for the design of quantum algorithms 
for search problems.
Magniez, Santha and Szegedy \cite{Magniez+SICOMP07}, using  
quantum walk search,  
constructed a quantum algorithm for triangle finding with improved query complexity 
$\tilde O(n^{13/10})$.

Besides Grover search and quantum walks, 
a third technique to design quantum query algorithms appeared recently when 
Reichardt \cite{ReichardtFOCS09} proved that the general adversary bound, initially 
shown to be only a lower bound on the quantum query complexity \cite{Hoyer+STOC07},
is actually an upper bound, which implies that the quantum query complexity of a problem
can be found by solving a semi-definite positive program. While this optimization problem 
in general exponentially many constraints, Belovs \cite{BelovsSTOC12}
then developed a technique known as the learning graph approach to restrict the search 
space to candidates that automatically satisfy the constraints, thus giving an intuitive and 
efficient way to obtain a (not necessarily optimal) solution of the original optimization problem.
Belovs \cite{BelovsSTOC12} illustrated the power of this new technique by using it to improve 
the quantum query complexity
of triangle finding to $O(n^{35/27})=O(n^{1.296...})$.
Lee, Magniez and Santha \cite{Lee+SODA13} then 
showed, again using learning graphs, how to further improve this query complexity to 
 $O(n^{9/7})=\tilde O(n^{1.285...})$, which was the best upper bound on the 
 quantum complexity of triangle finding known before the present paper.
These two results based on learning graphs actually used a
simple notion of learning graphs (referred to as ``non-adaptive'' learning graphs in \cite{Belovs+CCC13})
where the queries done by the algorithm do not depend on the values of prior queries, which implies that the 
same upper bound $O(n^{9/7})$ holds for weighted versions of the triangle finding problem as well.
Jeffery, Kothari and Magniez~\cite{Jeffery+SODA13} showed how this complexity can also be 
achieved, up to polylogarithmic factors, using quantum walks by introducing the concept 
of nested quantum walks.

The best known lower bound on the quantum query complexity of triangle finding is the trivial $\Omega(n)$.
Belovs and Rosmanis \cite{Belovs+CCC13} recently showed that 
any quantum algorithm (i.e., not necessarily based on learning graphs) solving the
edge-weighted triangle finding problem requires
$\Omega (n^{9/7}/\sqrt{\log n})$ queries. Since a non-adaptive learning graph does not treat differently
the unweighted triangle finding problem and its weighted versions, as mentioned above, 
this lower bound for the weighted case implies that 
any quantum algorithm for unweighted triangle finding constructed using a non-adaptive learning graph
requires $\Omega (n^{9/7}/\sqrt{\log n})$ queries as well, which matches, up to logarithmic factors, the best known upper bound described in the previous paragraph.
Practically, this means that, in order 
to improve by more than a $1/\sqrt{\log n}$ factor 
the $O(n^{9/7})$-query upper bound on the quantum query complexity  
of triangle finding, one need to take in consideration the difference between the unweighted triangle finding problem 
and its edge-weighted version. Moreover, if the learning graph approach is used, then the learning graph
constructed must be adaptive.
While a concept of adaptive learning graph has been developed by Belovs  
and used to design a new quantum algorithm for the $k$-distinctness problem \cite{BelovsFOCS12}, 
so far no application of this approach to the triangle finding problem has been discovered.

\paragraph{Statement of our result.}
In this paper we show that it is possible to overcome the $\Omega (n^{9/7}/\sqrt{\log n})$ barrier, 
and obtain the following result.
\begin{theorem}\label{theorem:main}
There exists a quantum algorithm that, 
given as input the oracle of an unweighted graph~$G$ on $n$ vertices,
outputs a triangle of $G$ with probability at least $2/3$ if a triangle exists,
and uses $\tilde O(n^{5/4})$ queries to the oracle.
\end{theorem}
This result shows, for the first time, that in the quantum setting unweighted triangle finding is easier than its edge-weighted version, and thus sheds light on the fundamental difference between these two problems.
Indeed, while in the classical time complexity setting strong evidences exist suggesting that the unweighted version is easier (as already mentioned, the unweighted version is not harder than Boolean matrix multiplication while the exact edge-weighted version is 3SUM-hard \cite{PatrascuSTOC10,Vassilevska+SICOMP13}), Theorem~\ref{theorem:main}, combined with the lower bound by Belovs and Rosmanis \cite{Belovs+CCC13}, enables us to give a separation between the quantum query complexities of these two problems.

Naturally, our result exploits the difference between the triangle finding problem 
and its weighted versions. Our approach does not rely on learning graphs or 
nested quantum walks, the techniques that were used to obtain the previous
best known upper bound. Instead, it 
relies on combinatorial ideas that exploit the fact that the graph is unweighted, 
as needed in any attempt to break the $\Omega (n^{9/7}/\sqrt{\log n})$ barrier,
combined with Grover search, quantum search with variable costs \cite{Ambainis10},
and usual quantum walks over Johnson graphs. 
Our quantum algorithm is highly adaptive, in that all later queries depend 
on the results of the queries done in at a preliminary stage by the algorithm. 
This gives another example of separation between the query complexity obtained 
by adaptive quantum query algorithms and the best query complexity that can be 
achieved 
using non-adaptive learning graphs (which is $\Omega (n^{9/7}/\sqrt{\log n})$ for triangle finding,
as mentioned above), and thus
sheds light on
limitations of the non-adaptive learning graph approach for graph-theoretical problems
such as triangle finding. 

\section{Preliminaries}
In this section we introduce some of our notations, briefly describe
the notion of quantum query algorithms for problems over graphs and 
present standard algorithmic techniques for solving search 
problems in the quantum query complexity setting.
We assume that the reader is familiar with the basics of
quantum computation and refer to, e.g., 
\cite{Buhrman+TCS02} for a more complete treatment
of quantum query complexity. 

For any finite set $X$ and any $r\in\{1,\ldots,|X|\}$
we denote $\Ss(X,r)$ the set of all subsets of~$r$ elements 
of $X$.
Note that 
$|\Ss(X,r)|={|X| \choose r}$.
We will use the notation $\E(X)$ to represent $\Ss(X,2)$,
i.e., the set of unordered pairs of elements in $X$.

Let $G=(V,E)$ be an undirected and unweighted graph, where
$V$ represents the set of vertices and $E\subseteq\E(V)$ 
represents the set of edges. In the query complexity setting, 
we assume that $V$ is known, and that~$E$ can be accessed through 
a quantum unitary operation $\mathcal{O}_G$ defined as follows.
For any pair $\{u,v\}\in \E(V)$, any bit $b\in\{0,1\}$, and any 
binary string $z\in\{0,1\}^\ast$, the operation $\mathcal{O}_G$ 
maps the basis state 
$\ket{\{u,v\}}\ket{b}\ket{z}$ to the state
\[
\mathcal{O}_G\ket{\{u,v\}}\ket{b}\ket{z}=\left\{
\begin{array}{ll}
\ket{\{u,v\}}\ket{b\oplus 1}\ket{z}&\textrm{ if } \{u,v\}\in E,\\
\ket{\{u,v\}}\ket{b}\ket{z}&\textrm{ if } \{u,v\}\notin E,\\
\end{array}
\right.
\]
where $\oplus$ denotes the bit parity (i.e., the logical XOR).
We say that a quantum algorithm computing some property of $G$
uses $k$ queries if the operation $\mathcal{O}_G$, given as an oracle, is
called $k$ times by the algorithm.

We describe below three algorithmic techniques, which we will use in this paper, 
to solve search problems 
over graphs in the quantum query complexity setting: Grover search, 
Ambainis' generalization of Grover search for variable costs, and 
quantum search algorithms based on quantum walks.

\paragraph{Grover search.}
Let $\Sigma$ be a finite set of size $m$.
Consider a Boolean function $f_G\colon\Sigma\to \{0,1\}$ depending on $G$ and assume
that, for any $s\in \Sigma$, the value $f_G(s)$ can be computed using 
$t$ queries to $\mathcal{O}_G$.
The goal is to find some element $s\in \Sigma$ such that $f_G(s)=1$, if
such an element exists. 
This problem can be solved by  repeating Grover's standard search~\cite{GroverSTOC96}
a logarithmic number of times, and checking if a solution has been found. 
For any constant $c>0$, 
this quantum procedure (called Safe Grover Search in~\cite{Magniez+SICOMP07}) uses 
$O(t\sqrt{m} \log m)$ queries to $\mathcal{O}_G$, outputs an element 
$s\in\Sigma$ such that $f_G(s)=1$ with probability at least $1-1/m^c$  if such an element 
exists, and always rejects if no such element exists. 
The same bound can actually be obtained even if, for each $s\in \Sigma$, the 
value $f_G(s)$ obtained using $t$ queries to $\mathcal{O}_G$ is correct only 
with high (e.g., greater than $2/3$) probability~\cite{Hoyer+ICALP03}.

\paragraph{Variable costs quantum search.}
Let $\Sigma$ be again a finite set of size $m$.
Consider a Boolean function $f_G\colon\Sigma\to \{0,1\}$ and assume that,
for each $s\in \Sigma$, there exists a quantum algorithm $\mathcal{B}_s$ 
making queries to $\mathcal{O}_G$ that satisfies the following properties:
\begin{itemize}
\item
$\mathcal{B}_s$ uses at most $t_s$ queries to $\mathcal{O}_G$; 
\item
$\mathcal{B}_s$ outputs $f_G(s)$ with probability at least $2/3$.
\end{itemize}
The goal is again to find some element $s\in \Sigma$ such that $f_G(s)=1$, if
such an element exists. Note that Grover search 
would lead to a quantum algorithm with query complexity $O(t_{max}\sqrt{m} \log m)$,
where $t_{max}$ represents the maximal value of $t_s$ over $s\in \Sigma$.
Ambainis \cite{Ambainis10} has shown how to do better when the square root of the 
average of the squares of the costs 
is significantly less than $t_{max}$.
We state this result in the following theorem where, for simplicity,
we assume that both $m$ and $t_{max}$ are upper bounded by 
a polynomial of $n$ (the number of vertices in the graph).

\begin{theorem}{(\cite{Ambainis10})}\label{theorem:ambainis}
Assume that there exists a constant $c$ such that $m\le n^c$
and $t_{max}\le n^c$.
There exists a quantum algorithm that makes 
\[
\tilde O\left(
\sqrt{ \sum_{s\in\Sigma} t_s^2}
\right)
\]
queries to $\mathcal{O}_G$ and
finds, with probability at least $3/4$, an element $s\in\Sigma$ such that $f_G(s)=1$  if such an element
exists.
\end{theorem}
As shown in \cite{Ambainis10}, it is not necessary to know the costs $t_s$ to obtain the complexity
stated in this theorem. Note that, while the formal statement of this theorem in \cite{Ambainis10} assumes that
the algorithms~$\mathcal{B}_s$ always output the correct answers, the case we consider (where each $\mathcal{B}_s$ outputs the correct answer only with high probability) is explicitly 
treated in Section 5 of \cite{Ambainis10}.

\paragraph{Quantum walk search.}
We now describe quantum walk search. For concreteness, 
we will restrict ourself to quantum walks over Johnson graphs, 
since they will be sufficient to obtain our results. We refer to  
\cite{Magniez+SICOMP11} for a more detailed and general 
treatment of the concept of quantum walks.

We start by defining Johnson graphs.
\begin{definition}
Let $T$ be a finite set and $r$ be a positive integer such that 
$r\le |T|$. The Johnson graph $J(T,r)$ is the undirected graph with vertex
set $\Ss(T,r)$ where two vertices
$R_1,R_2\in\Ss(T,r)$ are connected if and only if $|R_1\cap R_2|=r-1$.
\end{definition}

We now describe the kind of search problems related to a graph $G$ of $n$ vertices 
given as an oracle~$\mathcal{O}_G$ that can be solved using a quantum walk over a Johnson graph.
Let $T$ be a finite set and $r$ be a positive integer such that $r\le |T|$. 
For simplicity, we will assume that there exists a constant $c$ such that $|T|\le n^c$.
Let $f_G\colon \Ss(T,r)\to\{0,1\}$ be a Boolean function depending on~$G$, 
and write $M_G=f_G^{-1}(1)$.
The goal is to decide whether $M_G$ is empty or not, i.e., 
whether there exists some $R\in\Ss(T,r)$ such that $f_G(R)=1$.
Note that the search problem considered here is defined by the function $f_G$
(or, equivalently, by $M_G$),
and the input of this search problem is the graph~$G$. Its query complexity
corresponds to the number of queries to $\mathcal{O}_G$ needed
to decide whether $M_G$ is empty or not.

The above search 
problem can be solved using a quantum walk over the Johnson graph $J(T,r)$.
A state of the walk 
will correspond to a vertex (i.e., to a set $A\in\Ss(T,r)$), and 
a data structure $D_G(A)$, which in general depends on $G$,
will be associated to each state $A$. We say that the state $A$ 
is marked if $A\in M_G$.
Three types of cost are associated with $D_G$, all measured in the number of queries to $\mathcal{O}_G$.
The setup cost $\mathsf{S}$ is the cost to set up the data structure, i.e.,  the 
number of queries needed to construct $D_G(A)$ for a given vertex $A\in\Ss(T,r)$.
The update cost $\mathsf{U}$ is the cost to update the data structure, i.e.,  the number
of queries needed to convert $D_G(A)$ into $D_G(A')$ for two given connected vertices $A$ and $A'$ of $J(T,r)$.
The checking cost $\mathsf{C}$ is the cost of checking with probability greater than $2/3$,
given $A\in\Ss(T,r)$ and $D_G(A)$, 
if $A$ is marked.

Let $\varepsilon>0$ be such that, for all graphs $G$, the inequality
\[
\frac{|M_G|}{|\Ss(T,r)|}\ge \varepsilon
\]
holds whenever $M_G\neq\emptyset$. Ambainis \cite{AmbainisSICOMP07} has 
shown that the quantum walk over $J(T,r)$ described above 
will find with high probability an element in $M_G$, if such an element exists, using a number of 
queries of order $\mathsf{S}+\frac{1}{\sqrt{\varepsilon}}\left(\sqrt{r}\times \mathsf{U}+\mathsf{C}\right)$,
see also \cite{Magniez+SICOMP11} for discussions and generalizations.
For later reference, we state this result as the following theorem.
\begin{theorem}[\cite{AmbainisSICOMP07,Magniez+SICOMP11}]\label{theorem:walks}
The quantum walk over the Johnson graph $J(T,r)$ has query complexity
\[
\tilde O\left(\mathsf{S}+\frac{1}{\sqrt{\varepsilon}}\left(\sqrt{r}\times \mathsf{U}+\mathsf{C}\right)\right)
\]
and finds, with probability at least $3/4$, an element in $M_G$ if such an element exists.
\end{theorem}

\section{Overview of our algorithm}\label{sec:overview}
In this section we give an outline 
of the main ideas leading to our new quantum algorithm for triangle 
finding. The algorithm is described in details, and its query complexity 
rigorously analyzed, in Section~\ref{section:main}.

Let $G=(V,E)$ denote the undirected and unweighted graph that is the 
input of the triangle finding problem, and write $n=|V|$. 
For any vertex $u\in V$, we 
denote
\[
N_G(u)=\{v\in V\:|\: \{u,v\}\in E\}
\] 
the set of neighbors of $u$.

The algorithm first takes a set $X\subseteq V$ consisting of $ \Theta(\sqrt{n}\log n)$ vertices chosen 
uniformly at random from $V$, and checks if there exists a triangle of $G$ with a vertex in~$X$.
This can be checked, with high probability,
using Grover's quantum search \cite{GroverSTOC96} in
\[
O\left(\sqrt{|X|\times |\E(V)|}\right)=\tilde O\left(n^{5/4}\right)
\]
queries. Define 
\[
S=\bigcup_{u\in X}\E(N_G(u)).
\]
If no triangle has been reported, 
we know that any triangle of~$G$
must have an edge in the set
$
\E(V)\setminus S.
$
Note that the above preliminary step has already been used in prior works, in particular related to 
the design of combinatorial algorithms
for Boolean matrix multiplication (e.g., \cite{Bansal+ToC12,Schnorr+RANDOM98}) and even in 
the design of the $\tilde O(n^{10/7})$-query quantum algorithm for triangle finding 
in \cite{Magniez+SICOMP07,Szegedy03}. 
We now explain how to check whether $\E(V)\setminus S$ contains an edge of a triangle or not,
which is the novel contribution of this paper.

For any set $Y\subseteq V$ and any $w\in V$, let us define the set 
$\mydelta{Y,w}\subseteq \E(Y)$ as follows:
\[
\mydelta{Y,w}=\E(Y\cap N_G(w))\setminus S.
\]
It is easy to see that, with high probability on the choice of~$X$, for any $\{v,v'\}\in \E(V)\setminus S$ 
the inequality 
\[
\left|\{w\in V\:|\:\{v,w\}\in E \textrm { and } \{v',w\}\in E\}\right|\le \sqrt{n}
\]
holds -- the preliminary step of the previous paragraph was done precisely to obtain this sparsity condition. 
This implies that, for a vertex $w$ taken uniformly at random in~$V$, the expected 
size of $\mydelta{V,w}$ is at most $n^{3/2}$ and, more generally, for a random set $Y\subseteq V$
the expected size of $\mydelta{Y,w}$ is at most $|Y|^2/\sqrt{n}$ (see Lemma \ref{lemma:sparse} in 
Section~\ref{section:main}).
In this section we will describe our algorithm in the following situation:
there exists a positive constant $c$ such that
\begin{equation}\label{cond:outline}
\left|\mydelta{Y,w}\right|
\le  \frac{c|Y|^2}{\sqrt{n}}
\hspace{3mm}
\textrm{ for any $Y\subseteq V$ and any $w\in V$}.
\end{equation}
This assumption considerably simplifies the problem, 
eliminating several difficulties that the final algorithm will need to deal with,
but still represents a situation sufficiently non-trivial to enable us to describe 
well the main ideas of our algorithm.

Remember that we now want to check if $\E(V)\setminus S$ contains an edge 
of a triangle. Our key observation is the following. Given 
a vertex $w\in V$ and a set $B\subseteq V$ of size $\ceil{\sqrt{n}}$ such that 
$\mydelta{B,w}$ is known, we can check if there exists a pair $\{v_1,v_2\}\in\E(B)\setminus S$ 
such that $\{v_1,v_2,w\}$ is a triangle of $G$ with 
\[
O\left(\sqrt{|\mydelta{B,w}|}\right)=
O\left(\sqrt{\frac{c|B|^2}{\sqrt{n}}}\right)=O(n^{1/4})
\] 
queries using Grover search and Condition (\ref{cond:outline}), 
since such $\{v_1,v_2\}$ exists 
if and only if $\mydelta{B,w}\cap E\neq\emptyset$.
The remarkable point here is that, 
if there were no sparsity condition on $\mydelta{B,w}$
then this search would require $\Theta(\sqrt{|B|^2})=\Theta(\sqrt{n})$ queries.
This improvement from $\sqrt{n}$ to $n^{1/4}$ is one of the main
reasons why we obtain an algorithm for triangle finding with 
query complexity $\tilde O(n^{5/4})$ instead of $O(n^{3/2})$
using straightforward quantum search. 
Note that this observation, even combined with the other ideas we describe below,
does not seem to lead to efficient classical  
algorithms for triangle finding or Boolean matrix multiplication due to the large cost
required to construct $\mydelta{B,w}$ -- this is why it has not been exploited prior to the present work.
One of our main contributions is indeed
to show that, in the query complexity setting, a quantum algorithm can perform
this construction efficiently.

As just mentioned,
the main difficulty when trying to exploit the above observation 
is that we not only want now to find 
a vertex $w$ and a set $B$ for which there exists 
$\{v_1,v_2\}\in\E(B)\setminus S$ 
such that $\{v_1,v_2,w\}$ is a triangle, we also need to 
construct the set $\mydelta{B,w}$,
which requires additional queries. 
To deal with this problem, we use a quantum walk over a Johnson graph, 
which enables us to
implement the construction of $\mydelta{B,w}$ concurrently to the search of $B$
and~$w$. By carefully analyzing the resulting
quantum walk algorithm, we can show that the improvement by a factor~$n^{1/4}$ 
described in the previous paragraph
is still preserved as long as we have enough
prior information about the set~$S$ when executing the quantum walk.

The difficulty now is that loading enough information about $S$ during the execution of the quantum
walk is too costly. Moreover, constructing $S$ before executing the quantum walk 
requires $\Theta(n^{3/2})$ queries, which is too costly as well. To solve this difficulty, 
we first search, using another quantum walk on another Johnson graph, a set $A\subseteq V$ of size $\ceil{n^{3/4}}$ 
such that 
\[
\left(\bigcup_{w\in V}\mydelta{A,w}\right)\cap E \neq\emptyset, 
\]
and concurrently 
construct the set $\E(A)\setminus S$.
We then do exactly  
as in the previous paragraph, but taking~$B$ as a subset of $A$ instead of as a subset of $V$. 
Since $\mydelta{B,w}$ can be created efficiently from the knowledge of 
$\E(A)\setminus S$, and $\E(A)\setminus S$ is available 
in the memory of the new quantum walk, the problem mentioned in the previous 
paragraph is solved.  By carefully designing the new quantum walk, 
we can show that its query complexity is sufficiently small.
As an illustration of this claim, observe that constructing the set $\E(A)\setminus S$ for 
a given set $A\subseteq V$ of size $\ceil{n^{3/4}}$, which will be done by the quantum walk during its setup stage,
can be implemented using 
\[
O\left(|A|\times |X|\right)=\tilde O(n^{5/4})
\] 
queries by checking if $\{u,v\}\in E$ for all 
$u\in A$ and all $v\in X$.

To summarize, at a high-level our strategy to check if $\E(V)\setminus S$ contains an edge 
of a triangle, and thus check if $G$ contains a triangle, can be described as the following four-level recursive procedure:
\begin{itemize}
\item[1.]
Search for a set $A\subseteq V$ of size $\ceil{n^{3/4}}$ 
such that $\left(\bigcup_{w\in V}\mydelta{A,w}\right)\cap E \neq\emptyset$,
while concurrently constructing $\E(A)\setminus S$, using a quantum walk;
\item[2.]
Search for a vertex $w\in V$ such that $\mydelta{A,w}\cap E\neq\emptyset$;
\item[3.]
Search for a set $B\subseteq A$ of size $\ceil{\sqrt{n}}$ 
such that $\mydelta{B,w}\cap E\neq\emptyset$,
while concurrently constructing $\mydelta{B,w}$,
using a quantum walk and the fact that $\E(A)\setminus S$ has already been constructed;
\item[4.]
Check if $\mydelta{B,w}\cap E\neq\emptyset$ in
$O(n^{1/4})$ queries, using the fact that $\mydelta{B,w}$ has already been constructed.
\end{itemize}
  
Several technical difficulties arise when analyzing the performance
of this recursive quantum algorithm and showing that its query complexity
is $\tilde O(n^{5/4})$, especially when 
Condition~(\ref{cond:outline}) does not hold. They are dealt with by  
using additional quantum techniques, such as quantum search with variable costs, 
estimating the size of the 
involved sets by random sampling, 
and proving several concentration bounds.
Note that the order of the four levels of recursion in our algorithm is crucial
to guarantee the $\tilde O(n^{5/4})$ query complexity, and it does 
not seem that allowing further nesting in the quantum walks  
(e.g., using the recent concept of quantum nested walk
\cite{Belovs+ICALP13,Jeffery+SODA13}) can be used to 
further reduce the query complexity of our approach.

\section{Quantum Algorithm for Triangle Finding}\label{section:main}
In this section we prove Theorem \ref{theorem:main} by describing our quantum algorithm for triangle finding.

As in Section \ref{sec:overview},  $G=(V,E)$ will denote the undirected and unweighted graph that is the 
input of the triangle finding problem, and we write $n=|V|$. For any sets $X,Y\subseteq V$, we define the set $\mydelta{Y}\subseteq \E(Y)$
as follows:
\[
\mydelta{Y}=\E(Y)\setminus \bigcup_{u\in X}\E(N_G(u)),
\]
where $N_G(u)$ again denotes the set of neighbors of $u$.
As in Section \ref{sec:overview},
for any sets $X,Y\subseteq V$ and any vertex $w\in V$, we define the set $\mydelta{Y,w}\subseteq \mydelta{Y}$
as follows:
\begin{align*}
\mydelta{Y,w}&=\E(Y\cap N_G(w))\setminus \bigcup_{u\in X}\E(N_G(u))\\
&=\Big\{\{u,v\}\in \mydelta{Y}\:|\: \{u,w\}\in E \textrm { and } \{v,w\}\in E\Big\}.
\end{align*}

\subsection{Main algorithm and proof of Theorem \ref{theorem:main}}

A key combinatorial property related to the triangle finding problem 
that we use in this paper is highlighted in the following definition
of \emph{$k$-good sets}.

\begin{definition}
Let $k$ be any constant such that $0\le k\le 1$. 
A set $X\subseteq V$ is $k$-good for $G$ if 
the inequality
\[
\sum_{w\in V}\left|\mydelta{Y,w}\right|\le |Y|^2 n^{1-k}
\]
holds for all $Y\subseteq V$.
\end{definition}

Our algorithm will rely on the following observation, which shows 
that $k$-good sets for $G$ can be constructed very easily.
\begin{lemma}\label{lemma:sparse}
Let $k$ be any constant such that $0\le k\le 1$. Suppose that $X$ is a set obtained 
by taking uniformly at random, with replacement,  $\ceil{3n^{k}\log n}$ elements from $V$.
Then $X$ is $k$-good for $G$ with probability at least $1-1/n$.
\end{lemma}
\begin{proof}
Consider a pair $\{u,v\}\in\E(V)$ such that  
\[
\left|\{w\in V\:|\:\{u,w\}\in E \textrm { and } \{v,w\}\in E\}\right|> n^{1-k}.
\]
Let us write $T=\{w\in V\:|\:\{u,w\}\in E \textrm { and } \{v,w\}\in E\}$.
This pair is contained in $\mydelta{V}$ if and only if
$T \cap X= \emptyset$, which happens with probability
\[
\left(1-\frac{|T|}{n}\right)^{\ceil{3 n^k \log n}}<
\left(1-\frac{1}{n^k}\right)^{3 n^k\log n}\le \frac{1}{n^3}.
\]
By the union bound this implies that
with probability at least $1-\frac{1}{n}$
the inequality 
\[
\left|\{w\in V\:|\:\{u,w\}\in E \textrm { and } \{v,w\}\in E\}\right|\le n^{1-k}
\]
holds for all $\{u,v\}\in\mydelta{V}$.
The statement of the lemma then follows from a straightforward 
counting argument.
\end{proof}

In Section \ref{subsec:prop} we will prove the following proposition.
\begin{proposition}\label{proposition:checking}
Let $a$ and $k$ be two constants such that $0<a,k<1$.
Let $X$ be a known subset of~$V$ of size at most 
$\ceil{3n^{k}\log n}$ that is $k$-good for $G$.
There exists a quantum algorithm with query complexity
\[
\tilde O\left(
n^{1/2+k}+n^{1/2+2a/3}+n^{1/2+a-k/2}
\right)
\]
that, given as input 
a set 
$A\in\Ss(V,\ceil{n^a})$ and the set $\mydelta{A}$,
checks with probability at least $2/3$ if
$\mydelta{A}$ contains an edge of a triangle of~$G$.
\end{proposition}
Proposition \ref{proposition:checking} shows the existence of 
a quantum algorithm 
that checks efficiently if a known set $\mydelta{A}$ contains an
edge of a triangle of $G$, under the assumption that $X$ is 
$k$-good for $G$. With this result available, 
we are now ready to construct our $\tilde O(n^{5/4})$-query quantum algorithm 
for triangle finding.

\begin{proof}[Proof of Theorem \ref{theorem:main}]
Let $a$ and $k$ be two constants such that $0<a,k<1$. 
The values of these constants will be set later.

We first take a set $X\subseteq V$ obtained 
by choosing uniformly at random, with replacement,  $\ceil{3n^{k}\log n}$ elements from $V$,
and check if there exists a triangle of $G$ with a vertex in $X$.
This can be done
using Grover search with
\[
O\left(\sqrt{|X|\times |\E(V)|}\right)=\tilde O\left(n^{1+k/2}\right)
\]
queries. If no triangle has been reported, 
we know that any triangle of $G$
must have an edge in $\mydelta{V}$.
 
We now describe a quantum algorithm that finds a triangle with an edge in $\mydelta{V}$,
if such a triangle exists.
The idea is to search for a set $A\in\Ss(V,\ceil{n^a})$ such that
$\mydelta{A}$ contains an edge of a triangle. Once such a set 
$A$ has been found, a triangle can be found in 
\[
O\left(\sqrt{|V|\times |\E(A)|}\right)= O\left(n^{1/2+a}\right)
\]
queries using Grover search.
To find such a set $A$, we perform a quantum walk over the Johnson graph
$J(V,\ceil{n^a})$.
The states of this walk correspond to the elements in $\Ss(V,\ceil{n^{a}})$.
The state corresponding to a set $A\in \Ss(V,\ceil{n^{a}})$ is  
marked if $\mydelta{A}$ contains an edge of a triangle of $G$. 
In case the set of marked states is not empty, which means that
there exists $\{v_1,v_2\}\in \mydelta{V}$ that is an edge 
of a triangle of $G$,
the fraction of marked states is 
\[
\varepsilon\ge\frac{{n-2 \choose \ceil{n^a}-2}}{{n \choose \ceil{n^a}}}=\Omega\left(n^{2(a-1)}\right).
\]
In our walk, the data structure stores the set $\mydelta{A}$.
Concretely, this is done by 
storing the couple $(v,N_G(v)\cap X)$ for each $v\in A$, since this information is enough to construct
$\mydelta{A}$ without using any additional query.
The setup cost is 
$\mathsf{S}= |A|\times |X|=\tilde O(n^{a+k})$ queries.
The update cost is $\mathsf{U}=2|X|=\tilde O(n^{k})$ queries. From Theorem~\ref{theorem:walks},
the query complexity of our quantum walk is thus
\[
\tilde O\left(n^{a+k}+\sqrt{n^{2(1-a)}}\left(\sqrt{n^{a}}\times n^k+\mathsf{C}\right)\right),
\]
where $\mathsf{C}$ is the cost of checking if a state is marked.
Under the assumption that the set $X$ is $k$-good for~$G$,
Proposition~\ref{proposition:checking} shows that
\[
\mathsf{C}=\tilde O\left(
n^{1/2+k}+n^{1/2+2a/3}+n^{1/2+a-k/2}
\right).
\]
Note that Proposition \ref{proposition:checking} can be applied here since 
the set $\mydelta{A}$ is stored in the data structure, and thus known.
The query complexity of the quantum walk then becomes
\[
\tilde O\left(
n^{a+k}+n^{1-a/2+k}+
n^{3/2}
\left(
n^{k-a}+
n^{-a/3}+
n^{-k/2}
\right)
\right).
\]

Under the assumption that the set $X$ is $k$-good for~$G$,
the query complexity of the whole algorithm is thus
\[
\tilde O
\left(
n^{1+k/2}+n^{1/2+a}+
n^{a+k}+n^{1-a/2+k}+
n^{3/2}
\left(
n^{k-a}+
n^{-a/3}+
n^{-k/2}
\right)
\right).
\]
When the set $X$ is not $k$-good for~$G$, 
the algorithm may need more
queries to finish, but we simply stop immediately
when the number of queries exceeds the above upper bound, 
and in this case output, for instance, that $G$ does not contain any triangle. 
This decision may be wrong, but
Lemma \ref{lemma:sparse} ensures that this happens only with probability
at most $1/n$.

Finally, taking $a=\frac{3}{4}$ and $k=\frac{1}{2}$ gives query complexity $\tilde O(n^{5/4})$,
as claimed. 
\end{proof}

\subsection{Proof of Proposition \ref{proposition:checking}}\label{subsec:prop}
This subsection is devoted to proving Proposition \ref{proposition:checking}.

The quantum algorithm of Proposition \ref{proposition:checking} will use
quantum walks in which the query complexity of the checking procedures 
depends on the size of $\mydelta{A,w}$. To control the query complexity 
of these quantum walks, we will first need, given $A$, $X$ and $w$,
to estimate $|\mydelta{A,w}|$. This will be done using the 
classical algorithm described in the following lemma.

\begin{lemma}\label{lemma:estimation}
Let $A$ and $X$ be two subsets of $V$, 
and assume that $\mydelta{A}$ is known.
Let $m$ be a positive integer.
There exists a classical deterministic algorithm $\mathcal{A}$
with query complexity $O(m\log n)$,
which receives as input a binary string $s$ of length $\poly(m,\log n)$ and
a vertex $w\in V$, and outputs a real number $\mathcal{A}(s,w)$ satisfying the following condition: 
for a fraction at least $1-\frac{3}{n}$ of the strings $s$,
the inequalities
\begin{equation*}
\frac{1}{3}\times|\mydelta{A,w}|\le \mathcal{A}(s,w)\le \frac{3}{2}\times \max \left(\frac{|A|\times (|A|-1)}{2m},|\mydelta{A,w}|\right)
\end{equation*}
hold for all vertices $\omega\in V$.
\end{lemma}

While Lemma \ref{lemma:estimation} is proved by using relatively simple sampling arguments,
we mention some subtle points about its statement before proving it.
\begin{itemize}
\item
When $|\mydelta{A,w}|$ is too small, since giving a multiplicative estimation 
would require too many queries, we only ask that the output is upper bounded 
by $3(|A|(|A|-1))/(4m)$ for some parameter $m$ that can be chosen freely.
\item
While Lemma \ref{lemma:estimation} is proved by constructing a randomized algorithm based on random sampling,
Algorithm $\mathcal{A}$ in the statement of the lemma is a deterministic 
algorithm that receives a string~$s$ of polynomial length, 
intended to be the string of random bits used for sampling.
Later analyses will be considerably simplified by this formulation,
since the output of Algorithm $\mathcal{A}$ will be 
used, as already mentioned, to control the running times of the quantum walks we construct in 
Proposition \ref{proposition:checking}
(more complicated arguments would be necessary if these running times were random variables).
\item
A quantum algorithm based on quantum counting \cite{Brassard+ICALP98} could actually 
be used instead
of the classical algorithm~$\mathcal{A}$. While this would reduce the query complexity in Lemma \ref{lemma:estimation}, 
this does not reduce the final query complexity of our triangle finding algorithm.
\end{itemize}
We now proceed to the proof.

\begin{proof}[Proof of Lemma \ref{lemma:estimation}]
Consider the randomized algorithm $\mathcal{A}'$ described in Figure \ref{fig:algorithmA}.
This algorithm receives as input a vertex $w\in V$ and outputs a real number $\mathcal{A}'(w)$.
We define $\mathcal{A}$ as the deterministic version of $\mathcal{A}'$ where 
the bit flips used by $\mathcal{A}'$ are given to $\mathcal{A}$ as the additional input $s$.

\begin{figure}[ht!]
\begin{center}
\fbox{
\begin{minipage}{15 cm} 
\underline{\textbf{Algorithm $\mathcal{A'}$}}\vspace{2mm}

Input: a vertex $w\in V$
\begin{itemize}
\item[1.]
Initialize a counter $c_1$ to zero and then
repeat the following $\ceil{240\log n}$ times:
\begin{itemize}
\item[1.1.]
Take $m$ elements $\{u_1,v_1\},\ldots,\{u_m,v_m\}$ uniformly at random, with replacement, from $\E(A)$;
\item[1.2.]
Increment $c_1$ by one if 
there exists at least one index $i\in\{1,\ldots,m\}$ satisfying the following three conditions:
$\{u_i,v_i\}\in \mydelta{A}$ and
$\{u_i,w\}\in E$ 
and $\{v_i,w\}\in E$;
\end{itemize}
\item[2.]
If $c_1\le \ceil{240 \log n}/2$ then output $\mathcal{A}'(w)=\frac{|\E(A)|}{m}$;
\item[3.]
If $c_1> \ceil{240 \log n}/2$ then do:
\begin{itemize}
\item[3.2.] 
Initialize a counter $c_2$ to zero and then
repeat the following $\ceil{72m\log n}$ times:
\begin{itemize}
\item[3.2.1]
Take a pair $\{u,v\}$ uniformly at random from $\E(A)$;
\item[3.2.2]
Increment $c_2$ by one if the following three conditions are satisfied: $\{u,v\}\in \mydelta{A}$ and
$\{u,w\}\in E$ 
and $\{v,w\}\in E$;
\end{itemize}
\item[3.3]
Output $\mathcal{A}'(w)=\frac{c_2|\E(A)|}{\ceil{72m\log n}}$;
\end{itemize}
\end{itemize}
\end{minipage}
}
\end{center}\vspace{-4mm}
\caption{Algorithm $\mathcal{A'}$ computing an estimation of $|\mydelta{A,w}|$.}\label{fig:algorithmA}
\end{figure}

Note that only Steps 1.2 and 3.2.2 of Algorithm $\mathcal{A}'$ have non-zero query complexity.
Membership in $\mydelta{A}$ can be checked without query (since the set $\mydelta{A}$ is known), which implies
that the overall query complexity 
of $\mathcal{A}'$ is $O(m \log n)$. We show below that, for each vertex $w\in V$, the real number
$\mathcal{A'}(w)$ output by the algorithm satisfies
\begin{equation}\label{cond2}
\frac{1}{3}\times|\mydelta{A,w}|\le \mathcal{A}'(w)\le \frac{3}{2}\times \max \left(\frac{|\E(A)|}{m},|\mydelta{A,w}|\right)
\end{equation}
with probability at least $1-3/n^2$. The union bound then implies that, with probability at least $1-3/n$,
Condition (\ref{cond2}) holds for all $w\in V$,  
which concludes the proof.

Assume first that 
$|\mydelta{A,w}|< \frac{|\E(A)|}{3m}$.
Then the probability that $c_1$ is incremented by one during 
one execution of the loop of Steps 1.1-1.2 is 
\[
1-\left(1-\frac{|\mydelta{A,w}|}{|\E(A)|}\right)^{m}<
1-\left(1- \frac{1}{3m}\right)^{m}<
1-e^{-1/2}<0.4.
\]
From Chernoff bound, the inequality $c_1<\ceil{240\log n}/2$ holds at the end of the 
loop of Step 1
with probability at least
\[
1-\exp\left(-\frac{1}{3}\times \frac{1}{16}\times 0.4\ceil{240\log n}\right)\ge 1-\frac{1}{n^2}.
\]
When this happens, the algorithm outputs $\mathcal{A}'(w)=|\E(A)|/m$, which 
satisfies Condition~(\ref{cond2}).

Next, assume that $\frac{|\E(A)|}{3m}\le|\mydelta{A,w}|\le \frac{3|\E(A)|}{m}$.
In case the algorithm passes the test of Step~2, the output is $|\E(A)|/m$, which
satisfies Condition (\ref{cond2}). Otherwise,
since the probability that a pair $\{u,v\}$ taken uniformly 
at random in $\E(A)$ satisfies the conditions of Step 3.2.2 is 
$|\mydelta{A,w}|/|\E(A)|$, Chernoff bound implies that the output $\mathcal{A}'(w)=\frac{c_2|\E(A)|}{\ceil{72m\log n}}$
at Step~3.3 is between $\frac{1}{2}\times |\mydelta{A,w}|$
and $\frac{3}{2}\times |\mydelta{A,w}|$
with probability at least
\[
1-2
\exp\left(-\frac{1}{3}\times\frac{1}{4}\times \frac{\ceil{72m\log n}}{|\E(A)|}\times |\mydelta{A,w}|\right)
\ge 1- \frac{2}{n^2},
\]
in which case Condition (\ref{cond2}) is satisfied.

Finally, assume that $|\mydelta{A,w}|> \frac{3|\E(A)|}{m}$.
Then the probability that $c_1$ is not incremented during one execution
of the loop of Steps 1.1-1.2 is
\[
\left(1-\frac{|\mydelta{A,w}|}{|\E(A)|}\right)^{m}<
\left(1-\frac{3}{m}\right)^m\le e^{-3}<0.1.
\]
From Chernoff bound, the inequality $c_1>\ceil{240\log n}/2$ holds 
at the end of the loop of Step 1
with probability at least 
\[
1-\exp\left(-\frac{1}{2}\times\frac{16}{81}\times0.9\ceil{240\log n}\right)>1-\frac{1}{n^2},
\]
and then the algorithm proceeds to Step 3. If this happens then,
from the same argument as in the previous paragraph, 
the output of the algorithm is between $\frac{1}{2}\times |\mydelta{A,w}|$
and $\frac{3}{2}\times |\mydelta{A,w}|$
with probability at least $1-2/n^2$, in which case Condition~(\ref{cond2}) is satisfied.
\end{proof}

The following lemma will be used to give a lower bound on the fraction of marked states
in the quantum walks used by the quantum algorithm of Proposition \ref{proposition:checking}. 
\begin{lemma}\label{lemma:fraction}
Let $A$ and $X$ be two subsets of $V$, and assume that $|A|>3$.
Let $w$ be any vertex in $V$,
$\{v_1,v_2\}$ be any element of $\E(A)$, and 
 $r$ be an integer such that $3<r\le |A|$.
Suppose that $B$ is taken uniformly at random in $\Ss(A,r)$,
and consider the following two conditions:
\begin{itemize}
\item[(i)]
$\{v_1,v_2\}\in \E(B)$;
\item[(ii)]
$|\mydelta{B,w}|\le \frac{8(r-2)(r-3)}{(|A|-2)(|A|-3)}\times \frac{|\mydelta{A,w}|}{3}+16r$.
\end{itemize}
Then
\[
\Pr_{B\in \Ss(A,r)} \left[
\textrm{ Conditions (i) and (ii) hold }
\right]\ge
\frac{(r-1)^2}{2|A|^2}.
\]
\end{lemma}
\begin{proof}
First observe that 
\[
\Pr_{B\in \Ss(A,r)}\big[\textrm{ Condition (i) holds }\big]=
\frac{{|A|-2 \choose r-2}}{{|A| \choose r}}=
\frac{r(r-1)}{|A|(|A|-1)}\ge \frac{(r-1)^2}{|A|^2}.
\]
We show below that the inequality
\begin{equation*}
\Pr_{B\in \Ss(A,r)}\left[\textrm{ Condition (ii) does not hold } \:\big|\: \textrm{ Condition (i) holds }\right]\le
\frac{1}{2},
\end{equation*}
which will conclude the proof of the lemma.

Choosing $B$ under the assumption that $\{v_1,v_2\}\in \E(B)$
is equivalent to choosing $r-2$ vertices from $A\setminus\{v_1,v_2\}$. 
Let us call these vertices $v_3,\ldots,v_{r}$.
For each $\{i,j\}\in\E(\{1,\ldots,r\})$,
let $Y_{ij}$ denote the random variable with value one if 
$
\{v_i,v_j\}\in \mydelta{A,w}
$
and value zero otherwise.
We have
\[
|\mydelta{B,w}|=\sum_{\{i,j\}\in \E(\{1,\ldots,r\})}Y_{ij}.
\]
Note that, for each $\{i,j\}\in \E(\{3,\ldots,r\})$, for  any $\{u,u'\}\in \mydelta{A\setminus\{v_1,v_2\},w}$
the probability that $\{v_{i},v_{j}\}=\{u,u'\}$ is $\frac{1}{|\E(A\setminus\{v_1,v_2\})|}$.
We can thus use the upper bounds 
\[
\left\{
\begin{array}{ll}
E[Y_{ij}]\le \frac{|\mydelta{A,w}|}{|\E(A\setminus\{v_1,v_2\})|} &\textrm{ if } \{i,j\}\in\E(\{3,\ldots,r\}),\\
E[Y_{ij}]\le 1 &\textrm{ if } \{i,j\}\in\E(\{1,\ldots,r\})\setminus\E(\{3,\ldots,r\}),
\end{array}
\right.
\]
to derive the following upper bound on the expectation of $|\mydelta{B,w}|$:
\begin{align*}
E\Big[|\mydelta{B,w}|\Big]&=\sum_{\{i,j\}\in \E(\{1,\ldots,r\})} E\left[Y_{ij}\right]\\
&\le |\E(\{3,\ldots,r\})| \times \frac{|\mydelta{A,w}|}{|\E(A\setminus\{v_1,v_2\})|}
+|\E(\{1,\ldots,r\})\setminus \E(\{3,\ldots,r\})|\\
&= \frac{(r-2)(r-3)}{(|A|-2)(|A|-3)} \times |\mydelta{A,w}|
+(2r-3)\\
&\le \frac{(r-2)(r-3)}{(|A|-2)(|A|-3)} \times |\mydelta{A,w}|
+2r\:.
\end{align*}
Finally, let us write $\delta=\frac{2r(|A|-2)(|A|-3)}{(r-2)(r-3)}$.
From Markov's inequality, we have
\begin{align*}
\Pr\left[\textrm{ Condition (ii) does not hold } \:\big|\: \textrm{ Condition (i) holds }\right]
&\le\frac{1}{8}\times \frac{|\mydelta{A,w}|+\delta}{\frac{|\mydelta{A,w}|}{3}+\delta}\\
&\le\frac{1}{8}\times \left(\frac{|\mydelta{A,w}|}{\frac{|\mydelta{A,w}|}{3}} + \frac{\delta}{\delta}\right)\\
&\le\frac{1}{2}\:,
\end{align*}
as claimed.
\end{proof}

We are now ready to give the proof of Proposition \ref{proposition:checking}.
\begin{proof}[Proof of Proposition \ref{proposition:checking}]
The algorithm first takes a sufficiently long binary string $s$ uniformly at random.
We will later apply Algorithm~$\mathcal{A}$ of Lemma \ref{lemma:estimation} with $m=\ceil{n^{k}}$,
using this binary string~$s$ as input.
From Lemma \ref{lemma:estimation} we know that, with probability at least $1-3/n$ 
on the choice of $s$, the 
following property holds:
\begin{equation}\label{cond1}
\frac{1}{3}\times |\mydelta{A,w}|\le 
\mathcal{A}(s,w)\le \frac{3}{2}\times \max \left(\frac{\ceil{n^a}(\ceil{n^a}-1)}{2\ceil{n^k}},|\mydelta{A,w}|\right) \:\:\textrm{ for all }w\in V.
\end{equation}
We will show below that, when Property (\ref{cond1}) holds, for any fixed vertex $w\in V$ the cost 
of checking if there exists a pair $\{v_1,v_2\}\in \mydelta{A}$ such that 
$\{v_1,v_2,w\}$ is a triangle of~$G$ is
\[
Q(w)=\tilde O\left(n^k+n^{2a/3}+n^{a-k/2}+\sqrt{|\mydelta{A,w}|}\right)
\]
queries. 
When Property (\ref{cond1}) holds,
the algorithm of Theorem~\ref{theorem:ambainis} then enables us to
to check, with probability at least~$3/4$, 
the existence of a pair $\{v_1,v_2\}\in \mydelta{A}$ that is an edge of
a triangle of~$G$ with query complexity
\begin{align*}
\tilde O\left( \sqrt{\sum_{w\in V} (Q(w))^2}\right)
&=\tilde O\left( \sqrt{\sum_{w\in V}\left(\left(n^{k}+n^{2a/3}+n^{a-k/2}\right)^2+ |\mydelta{A,w}|\right)}\right)\\
&=\tilde O\left( n^{1/2+k}+n^{1/2+2a/3}+n^{1/2+a-k/2}+\sqrt{\sum_{w\in V}|\mydelta{A,w}|}\right)\\
&=\tilde O\left( n^{1/2+k}+n^{1/2+2a/3}+n^{1/2+a-k/2}\right),
\end{align*}
where the last equality is obtained using the fact
that $X$ is $k$-good.

When Property (\ref{cond1}) does not hold, 
which happens with probability at most $3/n$,
the algorithm may need more
queries to finish, 
but we simply stop immediately
when the number of queries exceeds the above upper bound, 
and in this case
output that $\mydelta{A}$ does not contain an edge of a triangle of $G$. 
This decision may be wrong but, again,
this happens only with probability
at most $3/n$.

We now show how to obtain the claimed upper bound on $Q(w)$,
the query complexity of checking if there exists a pair
$\{v_1,v_2\}\in \mydelta{A}$ such that 
$\{v_1,v_2,w\}$ is a triangle of~$G$ when Property~(\ref{cond1}) holds.
We first use Algorithm~$\mathcal{A}$ with input $(s,w)$
to obtain
$\mathcal{A}(s,w)$. The cost of this step is $\tilde O(n^k)$ queries,
from Lemma \ref{lemma:estimation}.
We then perform a quantum walk over the Johnson graph $J(A,\ceil{n^{2a/3}})$.
The states of this walk correspond to the elements in $\Ss(A,\ceil{n^{2a/3}})$.
We now define the set of marked states of the walk.
The state corresponding to a set 
$B\in \Ss(A,\ceil{n^{2a/3}})$ is marked if $B$ satisfies the following two conditions:
\begin{itemize}
\item[(i)]
there exists a pair $\{v_1,v_2\}\in \mydelta{B,w}$ such that $\{v_1,v_2\}\in E$ (i.e., such that $\{v_1,v_2,w\}$ is a triangle of~$G$); 
\item[(ii)]
$|\mydelta{B,w}|\le \frac{8(\ceil{n^{2a/3}}-2)(\ceil{n^{2a/3}}-3)}{(\ceil{n^a}-2)(\ceil{n^a}-3)}\times \mathcal{A}(s,w)+16\ceil{n^{2a/3}}$.
\end{itemize}
Lemma~\ref{lemma:fraction} shows that, when Property (\ref{cond1}) holds and in case there exists a pair
$\{v_1,v_2\}\in \mydelta{A}$ such that $\{v_1,v_2,w\}$ is a triangle of~$G$,
 the fraction of marked states is 
 \[
 \varepsilon=\Omega\left(n^{2\left(\frac{2a}{3}-a\right)}\right)=\Omega\left(n^{-2a/3}\right).
\]
The data structure of the walk will store $\mydelta{B,w}$. Concretely, this is done by 
storing the couple $(v,e_v)$ for each $v\in B$, where $e_v=1$ if $\{v,w\}\in E$
and $e_v=0$ if $\{v,w\}\notin E$ (observe that this information is indeed enough to construct
$\mydelta{B,w}$ without using any additional query, since the set $\mydelta{A}$ is known).
The setup cost is $\ceil{n^{2a/3}}$ queries since it is sufficient to check if $\{v,w\}$ is an edge for all $v\in B$.
The update cost is $2$ queries. The checking cost is
\[
O\left(\sqrt{|\mydelta{B,w}|}\right)= O\left(\sqrt{n^{-2a/3}\times \mathcal{A}(s,w)+n^{2a/3}}\right)
\]
queries,
since Condition (ii) can be checked without query (since $\mydelta{B,w}$ is stored in the data structure)
and then Condition (i) can be checked by performing a Grover search over $\mydelta{B,w}$.
Theorem \ref{theorem:walks}, applied under the assumption that Property~(\ref{cond1}) holds, thus gives the upper bound
\begin{align*}
Q(w)&=\tilde O\left(n^k+n^{2a/3}+\sqrt{n^{2a/3}}\left(\sqrt{n^{2a/3}}\times 2 +\sqrt{n^{-2a/3}\times \mathcal{A}(s,w)+n^{2a/3}}\right)\right)\\
&=\tilde O\left(n^k+n^{2a/3}+n^{2a/3} +\sqrt{\mathcal{A}(s,w)}\right)\\
&=\tilde O\left(n^k+n^{2a/3}+n^{a-k/2}+\sqrt{|\mydelta{A,w}|}\right),
\end{align*}
as claimed.
\end{proof}

\section*{Acknowledgments}
The author is grateful to Harumichi Nishimura and Seiichiro Tani for stimulating discussions about the quantum complexity of triangle finding and several comments about this paper, 
and to three anonymous reviewers for their suggestions.
This work is supported by
the Grant-in-Aid for Young Scientists~(B)~No.~24700005 of the Japan Society for the Promotion of Science
and the Grant-in-Aid for Scientific Research on Innovative Areas~No.~24106009 of
the Ministry of Education, Culture, Sports, Science and Technology in Japan.


\end{document}